\documentclass[a4paper]{article}

\usepackage[utf8]{inputenc}

\usepackage{amsmath}
\usepackage{amsfonts}
\usepackage{graphicx}

\usepackage{mathtools}

\usepackage[vlined, ruled, linesnumbered]{algorithm2e}

\newcommand{\softO}{\widetilde{O}}
\renewcommand{\epsilon}{\varepsilon}

\usepackage{xspace}
\newcommand{\eft}{\textsc{EFT}\xspace}

\newcommand{\vmark}{\ensuremath{\ell}}

\usepackage{amsthm}
\newtheorem{proposition}{Proposition}
\newtheorem{lemma}{Lemma}
\newtheorem{theorem}{Theorem}
\newtheorem{corollary}{Corollary}

\title{Compact and Fast Sensitivity Oracles for Single-Source Distances\thanks{A preliminary version of this work was accepted at the 24th European Symposium on Algorithms (ESA'16).}}

\usepackage{hyperref}

\usepackage{authblk}

\author{Davide~Bil\`o}
\affil{Dipartimento di Scienze Umanistiche e Sociali, Università di Sassari, Italy. \texttt{davide.bilo@uniss.it}}

\author{Luciano~Gual\`a}
\affil{Dipartimento di Ingegneria dell'Impresa, Università di Roma ``Tor Vergata'', Italy. \texttt{guala@mat.uniroma2.it}}

\author{Stefano~Leucci}
\affil{Dipartimento di Informatica, ``Sapienza'' Università di Roma, Italy. \texttt{leucci@di.uniroma1.it}}

\author{Guido~Proietti}
\affil{Dipartimento di Ingegneria e Scienze dell'Informazione e Matematica, Università degli Studi dell'Aquila, Italy.
IASI, CNR, Roma, Italy. \texttt{guido.proietti@univaq.it}}

\date{}

\begin{document}

\maketitle

\begin{abstract}
Let $s$ denote a distinguished source vertex of a non-negatively real weighted and undirected graph $G$ with $n$ vertices and $m$ edges.
In this paper we present two efficient \emph{single-source approximate-distance sensitivity oracles}, namely \emph{compact} data structures which 
are able to \emph{quickly} report an approximate (by a multiplicative stretch factor) distance from $s$ to any node of $G$ following the failure of any edge in $G$.
More precisely, we first present a sensitivity oracle of size $O(n)$ which 
is able to report 2-approximate distances from the source in $O(1)$ time.
Then, we further develop our construction by building, for any $0<\epsilon<1$,  
another sensitivity oracle having size $O\left(n\cdot \frac{1}{\epsilon} \log \frac{1}{\epsilon}\right)$, and which is able to report a $(1+\epsilon)$-approximate distance from $s$ to any vertex of $G$ in $O\left(\log n\cdot \frac{1}{\epsilon} \log \frac{1}{\epsilon}\right)$ time.
Thus, this latter oracle is essentially optimal as far as size and stretch are concerned, and it only asks for a logarithmic query time.  
Finally, our results are complemented with a space lower bound for the related class of single-source \emph{additively-stretched} sensitivity oracles, which is helpful to realize the hardness of designing compact oracles of this type.
\end{abstract}

\section{Introduction}
The term \emph{distance oracle} was coined by Thorup and Zwick \cite{TZ05}, to emphasize the quality of a data structure that, despite its sparseness, is able to report very quickly provably good approximate distances between any pair of nodes in a graph. Indeed, it is well-known that in huge graphs the trade-off between time and space for \emph{exact} distance queries is a very critical issue: at its extremes, either we use a quadratic (unfeasible) space to reply in constant time, or we use a linear space to reply at an unsustainable large time. Thus, a wide body of literature focused on the problem of developing intermediate solutions in between these two opposite approaches, with the goal of designing more and more compact and fast oracles. This already complex task is further complicated as soon as edge or vertex failures enter into play: here, the oracle should be able to return (approximate) distances following the failure of some component(s) in the underlying graph, or in other words to be \emph{fault-tolerant}, thus introducing an additional overload to the problem complexity. This kind of oracle is also known as \emph{distance sensitivity oracle}. In this paper we focus our attention on a such challenging scenario, but we restrict our attention to the prominent case in which concerned distances are from a fixed source only, which is of special interest in several network-based applications.

\subsection{Related work.}
Let $s$ denote a distinguished source vertex of a non-negatively real weighted and undirected $n$-vertex and $m$-edge graph $G=(V(G),E(G),w)$.
For the sake of avoiding technicalities, we assume that $G$ is 2-edge-connected, although this assumption can be easily relaxed without affecting our results.
A \emph{single-edge-fault-tolerant} $\alpha$-\emph{single-source distance oracle} (EFT $\alpha$-SSDO in the following), with $\alpha \geq 1$, is a data structure that for any $v \in V(G)$ and any $e \in E(G)$ is able to return an estimate of the distance in $G-e$ (i.e., the graph $G$ deprived by $e$) between $s$ and $v$, say $d_{G-e}(s,v)$, within the range $[d_{G-e}(s,v),\alpha \cdot d_{G-e}(s,v)]$. The term $\alpha$ is a.k.a. the \emph{stretch factor} of the oracle.

A natural counterpart of such an oracle is an EFT $\alpha$-\emph{approximate shortest-path tree} ($\alpha$-ASPT), i.e., a subgraph of $G$ which, besides a SPT of $G$ rooted at $s$,  contains $\alpha$-stretched shortest paths from $s$ after the failure of any edge $e$ in $G$.
Such a structure is also known as a \emph{single-source EFT $\alpha$-spanner}. In some sense, a SSDO aims to convert in an explicit form the distance information that a corresponding ASPT may retain just in an implicit form, similarly to the process of maintaining in an $n$-size array
all the distances from the source induced by the paths of a corresponding SPT.
However, such a conversion process is far to be trivial in general and should be accomplished carefully, since the exploitation of the implicit information may introduce a dilatation in the final size of the oracle.

While the study of sensitivity oracles for all-pairs distances started right after the first appearance of \cite{TZ05}, the single-source case was faced only later. More precisely, in \cite{DBLP:journals/siamcomp/DemetrescuTCR08} it was first proven that if we aim at \emph{exact} distances, then $\Theta(n^2)$ space may be needed, already for undirected graphs and single edge failures, and independently of the query time. Then, in \cite{BK13} the authors build
in $O(m \log n + n \log^2 n)$ time a \emph{single-vertex-fault-tolerant} (VFT) 3-SSDO of size $O(n \log n)$ and with constant query time. In the same paper, for \emph{unweighted} graphs and for any
$\epsilon >0$, the authors build in
$O(m\sqrt{n/\epsilon})$ time a VFT $(1+\varepsilon)$-SSDO of size $O(\frac{n}{\varepsilon^3}+n \log n)$ and with constant query time. Both oracles are \emph{path reporting}, i.e., they are able to report the corresponding approximate shortest path from the source in time proportional to the path size. Moreover, as discussed in \cite{BGLP14}, in both oracles/spanners the log-term in the size can be removed if \emph{edge} failures are considered, instead of vertex failures. Finally, they can easily be transformed into corresponding E/VFT ASPTs having a same size and stretch. As far as this latter result is concerned, this was improved in \cite{BGLP14}, where it was given, for any (even non-constant) $\epsilon>0$, an E/VFT $(1+\varepsilon)$-ASPT of size $O(\frac{n \log n}{\varepsilon^2})$, without providing a corresponding oracle, though.

Summarizing, we therefore have the following state-of-the-art for EFT SSDOs: if we insist on having linear-size and constant query time, then a $(1+\varepsilon)$-stretch can be obtained only for unweighted graphs, while for weighted graphs the best current stretch is 3. Actually, this latter value can be reduced only by either paying a quadratic size (by storing for every $e \in E(G)$, the explicit distances from $s$ in $G-e$), or an almost linear size but a super-linear query time (by storing and then inspecting the structure provided in \cite{BGLP14}). So, the main open question is the following: can we develop a good space-time trade-off (ideally, linear space and constant query time) by guaranteeing a stretch less than 3 (ideally, arbitrarily close to 1)?
In this paper, we make significant progresses in this direction.

\subsection{Our results.}
Our main result is, for any arbitrary small $\epsilon>0$, the construction in $O(mn+n^2\log n)$ time and $O\left(m+n\cdot \frac{1}{\epsilon} \log \frac{1}{\epsilon}\right)$ space of an EFT $(1+\epsilon)$-SSDO
having size $O\left(n\cdot \frac{1}{\epsilon} \log \frac{1}{\epsilon}\right)$ and query time $O\left(\log n\cdot \frac{1}{\epsilon} \log \frac{1}{\epsilon}\right)$. 
Thus, when $\epsilon$ is constant w.r.t. $n$, we get close to the ideal situation we were depicting above: our oracle has linear space, stretch arbitrarily close to 1, and a logarithmic query time. Moreover, it is interesting to notice that size and query time have an almost linear dependency on $1/\epsilon$.

To the best of our knowledge, this is the first EFT SSDO guaranteeing a $(1+\epsilon)$-stretch factor on weighted graphs.
Interestingly, our construction is not obtained by the EFT $(1+\epsilon)$-ASPT of size $O(\frac{n \log n}{\varepsilon^2})$ given in \cite{BGLP14}, whose conversion to a same size-stretch trade-off oracle sounds very hard, and is instead based on a quite different approach. More precisely, to get our size and query time bounds, we select a subset of \emph{landmark} nodes of $G$, and for each one of them we store $O\left(\frac{1}{\epsilon} \log \frac{1}{\epsilon}\right)$  \emph{exact} post-failure (for an appropriate set of failing edges) distances from $s$. Then, when an edge $e$ fails and we want to retrieve an approximate distance from $s$ towards a fixed destination node $t$, we efficiently select with the promised query time a pivotal landmark node that actually sits on a path in $G-e$ from $s$ to $t$ whose length is within the bound.
Notice that such a path is not \emph{explicitly} stored in our oracle, so unfortunately we cannot return it in a time proportional to its size (besides the query time). In other words, our oracle is not inherently path-reporting, an we leave this point as a challenging open problem.

To get the reader acquainted with our technique, we first develop in $O(mn+n^2\log n)$ time and $O(m)$ space an EFT $2$-SSDO of size $O(n)$ and constant query time. This result is of independent interest, since it is the first EFT SSDO with both optimal size and query time having a stretch better than the long-standing barrier of 3.
In this other oracle, once again we select a subset of landmark nodes of $G$, but in this case, to get the promised stretch, we do not need to maintain explicitly any exact distances towards them. Rather, for the failure of an edge $e$ of $G$ and for a fixed destination node $t$, a structural property of 2-stretched post-failure paths will allow us to return the 2-approximate distance from $s$ by simply understanding whether there exists a pivotal landmark node associated with $t$. Actually, we show that such an association can be established by formulating a corresponding \emph{bottleneck vertex query} problem on a rooted tree, that can be answered in $O(1)$ time by using a linear-size efficient data structure developed in \cite{DemaineLW14}.
  
 Finally, in order to better appreciate the quality of our former oracle, we provide a lower bound on the bit size of any EFT $\beta$-\emph{additive} SSDO, i.e., an oracle which is able to report a distance from $s$ following an edge failure which is exact unless an additive term $\beta$. Notice that for weighted graphs, as in our setting, it only makes sense that such a $\beta$ is depending on the actual queried distance $d$.
Notice also that our linear-size EFT $(1+\epsilon)$-SSDO can be revised as an EFT $(\epsilon \cdot d)$-additive SSDO.
 So, a naturally arising question is: for a given $0<\delta \le1$, can we devise a compact EFT $(\epsilon \cdot d^{1-\delta})$-additive SSDO? We provide an answer in the negative, by showing a class of graphs for which a corresponding set of oracles of this sort would contain at least an element of $\Omega(n^2)$ bit size, regardless of its query time.

\subsection{Other related results.}
Besides the aforementioned related work on single-source distance sensitivity oracles, we mention some further papers on the topic.
For directed graphs with integer positive edge weights bounded by $M$, in \cite{GW12} the authors
show how to build efficiently in $\softO(M n^{\omega})$ time a \emph{randomized} EFT SSDO of size $\Theta(n^2)$ and with $O(1)$ query time,
where returned distances are exact w.h.p., and $\omega< 2.373$ denotes the matrix
multiplication exponent. As far as \emph{multiple} edge failures are concerned, in \cite{DBLP:conf/stacs/BiloG0P16}, for the failure of any set $F \subseteq E(G)$ of at most $f$ edges of $G$, the authors build in $O(f m\, \alpha(m,n)+fn \log^3 n)$ time an $f$-EFT $(2|F|+1)$-SSDO of size $O(\min\{m,fn\} \log^2 n)$, with a query time of $O(|F|^2 \log^2 n)$, and that is also able to report the corresponding path in the same time plus the path size. Notice that this oracle is obtained by converting a corresponding single-source $f$-EFT spanner having size $O(fn)$ and a same stretch. Notice also that if one is willing to use $O(m \log^2 n)$ space, such oracle will be able to handle any number of edge failures (i.e., up to $m$).
Recently in \cite{SIROCCO15}, the authors faced the special case of \emph{shortest-path failures}, in which the failure of a set $F$ of at most $f$ adjacent edges along any source-leaf path has to be tolerated. For this problem, they build
in $O(n(m+f^2))$ time, a $(2k-1)(2|F|+1)$-SSDO of size $O(kn \,f^{1+1/k})$ and constant query time, where $|F|$ denotes the size of the actual failing path, and $k\geq 1$ is a parameter of choice.
Moreover, for the special case of $f=2$, they give an ad-hoc solution, i.e., a 3-SSDO that can be built in $O(nm+n^2 \log n)$ time, has size $O(n \log n)$ and constant query time.

In the past, several other research efforts have been devoted to \emph{all-pairs distance oracles} (APDO) tolerating single/multiple edge/vertex failures. Quite interestingly, here $\softO(n^2)$-size exact-distance sensitivity oracles are instead known, as opposed to the $\Omega(n^2)$ lower bound for the single-source case. More precisely, in \cite{BK09} the authors built (on directed graphs) in $\softO(mn)$ time
a 1-E/VFT $1$-APDO of size $\softO(n^2)$ and with query time $O(1)$.
For two failures, in \cite{DP09} the authors built, still on directed graphs, a 2-E/VFT $1$-APDO of size $\softO(n^2)$ and with query time $O(\log n)$.
Concerning multiple-edge failures,  in \cite{CLPR10} the authors built, for any
integer $k \geq 1$, an $f$-EFT $(8k - 2)(f + 1)$-APDO of size $O(fk \, n^{1+1/k}
\log (n W))$, where $W$ is the ratio of the maximum to the minimum edge weight
in $G$, and with a query time of
$\softO(|F| \, \log \log d)$, where $F$ is the actual set of failing edges, and $d$ is the distance between the
queried pair of nodes in $G-F$.

As we said before, the natural counterpart of distance sensitivity oracles are the fault-tolerant spanners.
Due to space limitations, for this related topic we refer the reader to the discussion and the references provided in \cite{DBLP:conf/stacs/BiloG0P16}. However, it is worth mentioning that there is a line of papers on EFT ASPTs \cite{NPW03,Parter15,PP13,PP14}, that as we said are very close in spirit to EFT SSDOs.

Finally, we mention that there is a large body of literature concerned with the design of ordinary (i.e., fault-free) distance oracles, and an extensive recent survey on the topic is given in \cite{DBLP:journals/csur/Sommer14}.

\subsection{Notation}
For two given vertices $x$ and $y$ of an edge weighted graph $H$, we denote by $\pi_H(x,y)$ a shortest path between $x$ and $y$ in $H$ and we denote by $d_H(x,y)$ the total length of $\pi_H(x,y)$. For two given paths $P$ and $P'$ such that $P$ is a path between $x$ and $y$ and $P'$ is a path between $y$ and $z$, we denote by $P\circ P'$ the path from $x$ to $z$ obtained by concatening $P$ and $P'$.

Let $T$ be an SPT of $G$ rooted at $s$, and let $e=(u,v)$ be an edge of $T$. In the rest of the paper, we always assume that $u$ is closer to $s$ than $v$ w.r.t. the number of hops in $T$. Furthermore, we denote by $T_v$ the subtree of $T$ rooted at $v$. Finally, for a vertex $t \in T_v$, we denote by $A(t,e)=V(\pi_T(v,t))$ the set of \emph{living ancestors} of $T$, $t$ included, contained in $T_v$.

\section{The EFT \texorpdfstring{$2$-SSDO}{2-SSDO}}
In this section we describe our EFT $2$-SSDO with linear size and constant query time. Some of the ideas we develop here will be used in the next section, where we provide our main result.

For the rest of the paper, let $T$ be a fixed SPT of $G$ rooted at $s$ that is stored in our distance oracle. First of all, observe that if there is no edge failure or the edge that has failed is not contained in $T$, then, for any vertex $t$, our distance oracle can return the (exact) distance value $d_T(s,t)$ in constant time. This is the case also when the edge $e=(u,v)$ that has failed is contained in $T$, but $t$ is not a vertex of $T_v$.
Therefore, in the rest of this section, we describe only how our distance oracle computes an approximate distance from $s$ to $t$ in $G-e$ when the edge $e=(u,v)$ that has failed is contained in $T$ and the vertex $t$ is contained in the subtree $T_v$.

The following lemma describes a simple but still interesting property that we exploit as key ingredient in our oracle. Let $e=(u,v)$ be a failing edge, we define a special replacement path from $s$ to $t$ as follows: $P_e(t) = \pi_{G-e}(s,v) \circ \pi_G(v,t)$.

\begin{lemma}
	\label{lemma:2_apx_conditions}
	Let $e=(u,v)$ be a failing edge and $t \in V(T_v)$. At least one of the following conditions holds: (i) $d_{G-e}(s,t) \le w(P_e(t)) \le 2 d_{G-e}(s,t)$, (ii) $d_{G-e}(s,t) < 2 d_{G}(s,t)$.
\end{lemma}
\begin{proof}
	We assume that (ii) is false (i.e., $d_{G-e}(s,t) \ge 2 d_G(s,t)$) and we prove that (i) must hold. Indeed:
	\begin{multline*}	
		d_{G-e}(s,t) \le w(P_e(t)) = d_{G-e}(s,v) + d_G(v,t) \le d_{G-e}(s,t) + d_{G-e}(v, t) + d_G(v,t) \\
		 = d_{G-e}(s,t) + 2 d_G(v,t) \le d_{G-e}(s,t) + 2 d_G(s,t) \le 2 d_{G-e}(s,t).
	\end{multline*}
\end{proof}

Notice that the length of $P_e(t)$ is available in constant time once we store $O(n)$ distance values, namely $d_{G-e}(s,v)$ for each $e=(u,v) \in E(T)$. Hence, the challenge here is to understand when $w(P_e(t))$ provides a 2-approximation of the distance $d_{G-e}(s,t)$ and when we can instead return the value $2 d_{G}(s,t) \le 2 d_{G-e}(s,t)$ (observe that $2 d_{G}(s,t)$ could be in general smaller than $d_{G-e}(s,t)$). The idea of our oracle is that of selecting a subset of \emph{marked} vertices for which this information can be stored and retrieved efficiently and from which we can derive the same information for the other nodes.

To this aim, we now describe an algorithm that preprocesses the graph and collects compact information that we  will use later to efficiently answer distance queries. Consider the edges of $T$ as traversed by a preorder visit from $s$. We define a total order relation $\prec$ on $E(T)$ as follows: we say that $e' \prec e''$ iff $e'$ is traversed before $e''$. We also use $e' \preceq e''$ to denote that either $e' \prec e''$ or $e' = e''$.

\begin{algorithm}[t]
	\caption{Mark-up algorithm}	
	\label{alg:mark}
	
	\DontPrintSemicolon
	\For{$v \in V$}
	{
		$\vmark(v) \gets \infty$ \;
	}
	
	\BlankLine
	
	\For{$e=(u,v) \in E(T)$ in preorder w.r.t. $T$}
	{
		\For{$t \in V(T_v)$ in preorder w.r.t. $T$}
		{
			\If(\tcp*[f]{Distance test}){$w(P_e(t)) \le 2 d_{G-e}(s,t)$}
			{
				do nothing \;
			}
			\ElseIf(\tcp*[f]{Ancestor test}){$\exists z \in A(t,e) \, : \, \vmark(z) \neq \infty$}
			{
				do nothing \;
			}
			\Else(\tcp*[f]{Both tests failed})
			{
				$\vmark(t) \gets e$ \tcp*{Mark $t$ at time $e$}
			}
		}
	}
\end{algorithm}

Algorithm~\ref{alg:mark} considers the failing edges $e \in E(T)$ in preorder and computes a label $\vmark(v)$ for each vertex $v \in V(G)$. This value will be either $\infty$ or a suitable edge $e \in E(T)$. Here we treat $\infty$ as a special label that satisfies $e' \prec \infty$ for every edge $e' \in E(T)$.
We say that $v$ is \emph{marked} if $\vmark(v) \neq \infty$, and we say that $v$ is \emph{marked at time $e$} if $\vmark(v) \preceq e$. Intuitively, $\vmark(v)$ is the time at which $v$ first becomes marked.

More precisely, for each failing edge $e$, Algorithm~\ref{alg:mark}, marks a vertex $t \in V(T_v)$ (at time $e$) iff vertex $t$ fails two tests: the \emph{distance test} and the  \emph{ancestor test}. In the distance test we check whether the path $P_e(t)$ suffices to provide a $2$-stretched distance to $t$, while in the ancestor test we check whether a living ancestor of $t$ has already been marked. Notice that the ancestor test guarantees that each vertex $t$ is marked at most once during the whole execution of the algorithm (since $t \in A(t,e)$ by definition).

As a simple consequence of the above algorithm, we have:

\begin{lemma}
	\label{lemma:marked_bound}
	Let $e \in E(T)$ be a failing edge and let $t$ be a vertex such that $\vmark(t)=e$, we have $d_{G-e}(s,t) < 2 d_G(s,t)$.
\end{lemma}
\begin{proof}
	Since $t$ is first marked at time $e$, it must have failed the distance test, i.e., $w(P_e(t)) > 2 d_{G-e}(s,t)$.
	This means that condition (i) of Lemma~\ref{lemma:2_apx_conditions} is false and hence condition (ii) must hold.
\end{proof}

Another useful property of the marked vertices is the following:

\begin{lemma}
\label{lemma:safe_path}
Let $e \in E(T)$ be a failing edge and let $t$ be a vertex such that $\vmark(t)= e$, then $\pi_{G-e}(s,t)$ and $\pi_{T}(v,t)$ are edge disjoint.
\end{lemma}
\begin{proof}
	Let $e=(u,v)$ and assume by contradiction that $\pi_{G-e}(s,t)$ and $\pi_{T}(v,t)$ are not edge disjoint. Let $(z,z')$ be an edge belonging to both paths, with $z$ closer to $v$ than $z'$. Notice that both $z$ and $z'$ are living ancestors of $t$, and that $z \neq t$.
	
	Since $t$ is first marked at time $e$, it must have failed the ancestor test. This implies that no other living ancestor of $t$ is marked at time $e$. Moreover, as $z$ is visited by the algorithm before $t$, it must have failed the ancestor test as well.
	 Since $z$ it is not marked at time $e$, it follows that it must have passed the distance test, i.e., $w(P_e(z)) \le 2 d_{G-e}(s,z)$. We have $P_e(t) = P_e(z) \circ \pi_G(z,t)$ and hence:
	\begin{multline*}	
		w( P_e(t) )  = w(P_e(z)) + d_G(z,t) \le 2 d_{G-e}(s,z) + d_G(z,t) \\
		  \le 2 d_{G-e}(s,z) + 2 d_{G-e}(z,t) = 2 d_{G-e}(s,t)
	\end{multline*}
	\noindent which implies that $t$ has passed the distance test and contradicts the hypothesis $\vmark(t)=e$.
\end{proof}

The next lemma is the last ingredient of our oracle, and allows to distinguish the two cases of Lemma~\ref{lemma:2_apx_conditions}.

\begin{figure}[t]
	\centering
	\includegraphics[scale=1]{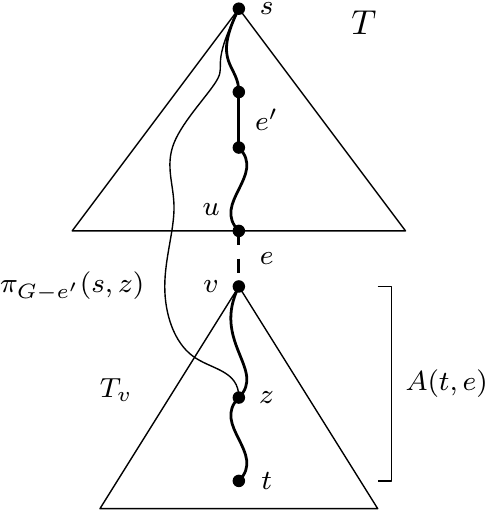}
	\caption{Representation of the proof of Lemma~\ref{lemma:ancestor_bound}. The shortest path between $s$ and $t$ in $T$ is shown in bold while the failing edge $e$ is dashed. Notice that the path $\pi_{G-e'}(s,z)$ is edge disjoint from the path $\pi_T(v,z)$.}
	\label{fig:ancestor}
\end{figure}

\begin{lemma}
	\label{lemma:ancestor_bound}
	Let $e = (u,v) \in E(T)$ be a failing edge and let $t \in V(T_v)$. If there exists $z \in A(t,e)$ such that  $\vmark(z) \preceq e$, then 	$d_{G-e}(s,t) \le 2 d_G(s,t)$. If no such vertex $z$ exists, then $d_{G-e}(s,t) \le w(P_e(t)) \le 2 d_{G-e}(s,t)$.
\end{lemma}
\begin{proof}
	Let $z$ be any vertex in $A(t,e)$ such that $\vmark(z) \preceq e$, and let $e'=\vmark(z)$. By the definition of living ancestor and by Lemma~\ref{lemma:safe_path} we have that $\pi_{G-e'}(s,z)$ does not use the edge $e$ (see Figure~\ref{fig:ancestor}). Since $z$ is marked at time $e'$ we have $d_{G-e'}(s,z)<2 d_G(s,z)$ (see Lemma~\ref{lemma:marked_bound}). Thus, we have that $d_{G-e}(s,z) \le w(\pi_{G-e'}(s,z))=d_{G-e'}(s,z)<2 d_G(s,z)$. Therefore:
	\begin{align*}	
		d_{G-e}(s,t) & \le d_{G-e}(s,z) + d_G(z,t) \le 2 d_G(s,z)  + d_G(z,t) \\
		& \le 2 d_G(s,z) + 2 d_G(z,t) = 2 d_G(s,t).
	\end{align*}
	
	If no such vertex $z$ exists, then when Algorithm~\ref{alg:mark} considered edge $e$, the vertex $z$ failed the ancestor test.
	Since $t$ is not marked at time $e$ (as otherwise we could choose $z=t$) it must have passed the distance test, i.e., $w(P_e(t)) \le 2 d_{G-e}(s,t)$.
\end{proof}

This latter lemma is exactly what we need in order to implement the query operation of our oracle. When edge $e=(u,v)$ is failing and we are queried for the distance of a vertex $t$, we first test whether $e \in E(T)$ and $t \in V(T_v)$: if the test fails we return the original distance $d_G(s,t)$.\footnote{To see whether $t$ is contained in $V(T_v)$ or not, it suffices to check whether the least common ancestor of $t$ and $v$ in $T$ corresponds to $v$ or not. The least common ancestor between any pair of vertices of a tree can be computed in constant time after a linear time preprocessing \cite{HarelT84}.} If the test succeeds, we look for a vertex $z \in A(t,e)$ such that $\vmark(z) \preceq e$. If such a vertex exists we return $2 d_G(s,t)$, otherwise we return $w(P_e(t))$. Observe that in both cases we return a feasible 2-approximation of the distance $d_{G-e}(s,t)$.

In the following we will show how it is possible to determine in constant time whether such a vertex $z$ exists.
More precisely we only need to look for a vertex $x \in A(t, e)$ minimizing $\vmark(x)$.
If such a vertex satisfies $\vmark(x) \preceq e$ then $z=x$ and we are done. On the converse, if $e \prec \vmark(x)$, then we know that no vertex $z \in A(t, e)$ with $\vmark(z) \preceq e$ can exist.

To this aim, we use a data structure for the \emph{bottleneck vertex query} problem on trees (\texttt{BVQ} for short). In the \texttt{BVQ} problem we want to preprocess a vertex-weighted tree $\mathcal{T}$ in order to answer queries of this form: given two vertices $x,y \in V(\mathcal{T})$ report the lightest vertex on the (unique) path between $x$ and $y$ in $\mathcal{T}$. In \cite{DemaineLW14}, the authors show how to build, in $O(|V(\mathcal{T})| \log |V(\mathcal{T})|)$ time, a data structure having linear size and constant query time.\footnote{Actually, in \cite{DemaineLW14} the \emph{bottleneck edge query} (\texttt{BEQ}) problem is considered instead. However it is easy to see that the \texttt{BEQ} and the  problems \texttt{BVQ} are essentially equivalent.}

In our preprocessing, we build such a structure on the tree $T$ where each vertex $x \in T$ weighs $\vmark(x)$, and then we use it to locate $x$ in the path between $v$ and $t$ whenever we need to report an approximate distance for $d_{G-(u,v)}(s,t)$.

We are now ready to state the main result of this section.

\begin{theorem}
Let $G$ be a non-negatively real weighted and undirected $n$-vertex and $m$-edge graph, and let $s$ be a source node.
There exists an EFT $2$-SSDO that has size $O(n)$ and constant query time, and that can be constructed using $O(mn+n^2 \log n)$ time and $O(m)$ space.
\end{theorem}
\begin{proof}
As we already discussed it is easy to answer a query in constant time once we store: (i) the SPT $T$ of $G$ w.r.t. $s$, (ii) the label $\vmark(v)$ for each $v$, (iii) the value $w(\pi_{G-e}(s,v))$ for each $(u,v) \in E(T)$, and (iv) a data structure for the \texttt{BVQ} problem. The total space used is hence $O(n)$.

Concerning the time and the space used by Algorithm \ref{alg:mark}, observe that for each edge $e=(u,v)$, we can compute an SPT of $G-e$ with source $s$ in $O(m + n \log n)$ time and $O(m)$ space. Therefore, for each $t$ the distance test can be accomplished in $O(1)$ time. It remains to show that also the ancestor test can be done in constant time. To this aim, it is sufficient to maintain for each vertex $x$ the (current) number $\nu_x$ of marked ancestors of $x$ in $T$, and check whether $\nu_t - \nu_u >0$. The maintenance  of these values can be clearly done with constant time and space overhead, from which the claim follows.
\end{proof}

\section{The EFT \texorpdfstring{$(1+\epsilon)$-SSDO}{(1+epsilon)-SSDO}}

In this section we describe our main result, namely how to build, given any $0<\epsilon<1$, an EFT $(1+\epsilon)$-SSDO having $O\left(n\cdot \frac{1}{\epsilon} \log \frac{1}{\epsilon}\right)$ size and $O\left(\log n\cdot \frac{1}{\epsilon} \log \frac{1}{\epsilon}\right)$ query time.

Our distance oracle stores a set of $O\left(n\cdot \frac{1}{\epsilon} \log \frac{1}{\epsilon}\right)$ (exact) distance values that are computed by a preprocessing algorithm that we describe below. From a high-level point of view, we follow the same approach used in the previous section, but here a vertex $t$ can be marked several times, each corresponding to a specific failing edge $e=(u,v) \in E(T)$ for which the algorithm computes the shortest path $\pi_{G-e}(s,t)$ that is edge disjoint from $\pi_{T}(v,t)$. We will show that such paths have strictly decreasing lengths and that they are $O\left(\frac{1}{\epsilon} \log \frac{1}{\epsilon}\right)$ in number. We will store all these distance values and we will show that they can be used to efficiently answer any distance query by suitably combining them with distances in $T$.

More precisely, for every $e=(u,v) \in E(T)$ and every $t \in V(T_v)$, the preprocessing algorithm computes a value ${\tt dist}(t,e)$ that satisfies $d_{G-e}(s,t)\leq {\tt dist}(t,e)\leq \sqrt{1+\epsilon}\cdot d_{G-e}(s,t)$. Furthermore, each value ${\tt dist}(t,e)$ represents the total length of a path $P$ from $s$ to $t$ in $G-e$, whose structure can be either of the following two types:
\begin{description}
\item [type 1:] $P=\pi_{G-e}(s,t)$;
\item [type 2:] $P$ can be decomposed into $\pi_{G-e'}(s,z)$, for some $e'$ and $z$ such that ${\tt dist}(z,e')=d_{G-e'}(s,z)$, and $\pi_T(z,t)$ (possibly, either $e=e'$ or $z=t$).
\end{description}

Since each path of type 2 can be easily derived by combining a path of type 1 with a path in $T$, our oracle stores only all the values ${\tt dist}(t,e)=d_{G-e}(s,t)$ that represent paths of type 1. In the next two subsections, we will show that, for every $e \in E(T)$ and every $t$, our distance oracle can compute a $(\sqrt{1+\epsilon})$-approximation of ${\tt dist}(t,e)$ in $O\left(\log n\cdot \frac{1}{\epsilon} \log \frac{1}{\epsilon}\right)$ time.

\subsection{The preprocessing algorithm}

The preprocessing algorithm (see the pseudocode of Algorithm \ref{alg:preprocessing}) visits all the edges of $T$ in preorder and, for each $e=(u,v) \in E(T)$, it visits all the vertices of $T_v$ in preorder. For the rest of this section, unless stated otherwise, let $e=(u,v)$ be a fixed edge of $T$ that is visited by the algorithm. The algorithm sets ${\tt dist}(v,e)=d_{G-e}(s,v)$, i.e., ${\tt dist}(v,e)$ always represents a path of type 1. When the algorithm visits $t$, with $t\neq v$, it first checks whether the shortest, among several paths from $s$ to $t$ in $G-e$ of type 2, has a total length of at most $\sqrt{1+\epsilon}\cdot d_{G-e}(s,t)$. If this is the case, then the algorithm sets ${\tt dist}(t,e)$ equal to the total length of such a path, otherwise it sets ${\tt dist}(t,e)=d_{G-e}(s,t)$, i.e., ${\tt dist}(t,e)$ represents a path of type 1. The preprocessing algorithm returns the set of all distance values that represent the paths of type 1.

For each vertex $t$, the algorithm stores the total length of the last path from $s$ to $t$ of type 1 that has computed in the variable ${\tt last}(t)$.

\begin{algorithm}
	\DontPrintSemicolon
	\SetKwInOut{Input}{Input}
	\SetKwInOut{Output}{Output}
	\SetKw{Break}{break}	

	\caption{Selects paths of type 1 whose lengths are stored in the oracle.}
		\label{alg:preprocessing}
	
	\tcp{Initialization of variables}
	$S,S'=\emptyset$
	\For{every $t \in V(G)$}
	{
		${\tt last}(t)=\infty$\;
	}

	\BlankLine

	\tcp{All the values ${\tt dist}(t,e)$ are computed}
	\For{every $e=(u,v) \in E(T)$ in preorder w.r.t. $T$}
	{
		${\tt last}(v),{\tt dist}(v,e)=d_{G-e}(s,v)$; add $d_{G-e}(s,v)$ to $S'$ \tcp*{path of type 1}
		\For{every $t \in V(T_v)\setminus\{v\}$ in preorder w.r.t. $T$}
		{
			\tcp{The length of a path from $s$ to $t$ in $G-e$ of type 2 is computed}
			${\tt dist}(t,e)=\min\big\{{\tt last}(z)+d_T(z,t) \mid z \in A(t,e)\big\}$\;
				\If{${\tt dist}(t,e)>\sqrt{1+\epsilon}\cdot d_{G-e}(s,t)$}
				{
					${\tt last}(t),{\tt dist}(t,e)=d_{G-e}(s,t);$ add $d_{G-e}(s,t)$ to $S$ \tcp*{path of type 1}
				}
		}
	}

	\BlankLine
	
	\Return $S$ and $S'$.
\end{algorithm}

For the rest of this section, unless stated otherwise, let $t$ be a fixed vertex of $T_v$ that is visited by the algorithm. The proof of the following proposition is trivial.

\begin{proposition}\label{prop:feasibility}
At the end of the visit of $t$, ${\tt dist}(t,e)\leq \sqrt{1+\epsilon}\cdot d_{G-e}(s,t)$.
\end{proposition}

\noindent The following lemma is similar to Lemma \ref{lemma:safe_path} and it is useful to prove that ${\tt dist}(t,e)\geq d_{G-e}(s,t)$.

\begin{lemma}\label{lm:disjoint_paths}
If $d_{G-e}(s,t)$ is added to $S\cup S'$, then $\pi_{G-e}(s,t)$ and $\pi_T(v,t)$ are edge disjoint.
\end{lemma}
\begin{proof}
The claim trivially holds when $t=v$ since $\pi_T(v,v)$ contains no edge. Therefore, we assume that $t\neq v$. We prove the claim by contradiction by showing that if $\pi_{G-e}(s,t)$ and $\pi_T(v,t)$ were not edge disjoint, then the algorithm would not add $d_{G-e}(s,t)$ to $S\cup S'$. So, we assume that $\pi_{G-e}(s,t)$ and $\pi_T(v,t)$ are not edge disjoint. Let $t'$ be, among the vertices that are contained in both $\pi_{G-e}(s,t)$ and $\pi_T(v,t)$, the one that is closest to $v$ w.r.t. the number of hops in $\pi_T(v,t)$. Clearly, $t'\neq t$ and $\pi_T(t',t)$ is a shortest path from $t'$ to $t$ in $G$ as well as in $G-e$. Thus, by the suboptimality property of shortest paths,
$d_{G-e}(s,t)=d_{G-e}(s,t')+d_{G-e}(t',t)=d_{G-e}(s,t')+d_T(t',t)$. Let $z \in A(t',e)$ be the vertex such that ${\tt dist}(t',e)={\tt last}(z)+d_T(z,t')$ (possibly $z=t'$). As the algorithm visits $t'$ before visiting $t$, by Proposition \ref{prop:feasibility}, ${\tt dist}(t',e)\leq \sqrt{1+\epsilon}\cdot d_{G-e}(s,t')$ at the beginning of the visit of $t$. Therefore
\begin{align*}
{\tt last}(z)+d_T(z,t)& = {\tt last}(z)+d_T(z,t')+d_T(t',t) = {\tt dist}(t',e)+d_T(t',t)\\
				& \leq \sqrt{1+\epsilon}\cdot d_{G-e}(s,t')+d_T(t',t)\leq \sqrt{1+\epsilon} \cdot d_{G-e}(s,t).
\end{align*}
As ${\tt dist}(t,e)\leq {\tt last}(z)+d_T(z,t)$ already before the execution of the if statement during the visit of $t$, the algorithm never adds $d_{G-e}(s,t)$ to $S\cup S'$. The claim follows.
\end{proof}

\noindent We now prove that ${\tt dist}(t,e)\geq d_{G-e}(s,t)$.

\begin{lemma}\label{lm:feasibility_bis}
At the end of the visit of $t$, ${\tt dist}(t,e)\geq d_{G-e}(s,t)$.
\end{lemma}
\begin{proof}
The claim trivially holds if the algorithm sets ${\tt dist}(t,e)=d_{G-e}(s,t)$. Therefore, we need to prove the claim when the condition of the if statement during the visit of $t$ is not satisfied, i.e., ${\tt dist}(t,e)={\tt last}(z)+d_T(z,t)$, for some vertex $z \in A(t,e)$ (possibly, $z=t$). Let ${\tt last}(z)=d_{G-e'}(s,z)$, for some $e'=(u',v')$ such that $z$ is a vertex of $T_{v'}$ (possibly $e'=e$). We divide the proof into the following two cases according to whether $e'=e$ or not.

Consider the case in which $e'=e$ and observe that $e$ is not contained in $\pi_T(z,t)$. Therefore ${\tt dist}(t,e)=d_{G-e'}(s,z)+d_T(z,t)=d_{G-e}(s,z)+d_{G-e}(z,t)\geq d_{G-e}(s,t)$.

Consider the case in which $e'\neq e$ and observe that $e$ is an edge of the path $\pi_T(v',z)$. Furthermore, ${\tt last}(z)=d_{G-e'}(s,z)$ implies that the algorithm has added $d_{G-e'}(s,z)$ to $S \cup S'$. Therefore, by Lemma \ref{lm:disjoint_paths}, $\pi_{G-e'}(s,z)$ and $\pi_T(v',z)$ are edge disjoint. This implies that $e$ is contained neither in $\pi_{G-e'}(s,z)$ nor in $\pi_T(z,t)$. Therefore, $d_{G-e}(s,t)\leq d_{G-e'}(s,z)+d_T(z,t)={\tt last}(z)+d_T(z,t)={\tt dist}(t,e)$, and the claim follows.
\end{proof}

The following proposition allows us to prove that the number of paths of type 1 computed by the algorithm is almost linear in $n$.

\begin{proposition}\label{prop:rank_of_a_node}
Let $e_0,e_1,\ldots,e_k$ be all the pairwise distinct edges of $T$, in the order in which they are visited by the algorithm, such that $d_{G-e_i}(s,t)\in S$. Then, for every $i=0,1,\ldots,k$, $d_{G-e_i}(s,t)<2/\big((\sqrt{1+\epsilon}-1)(1+\epsilon)^{i/2}\big)d_G(s,t)$. Furthermore, $k< 2\cdot \frac{\log\big(2/(\sqrt{1+\epsilon}-1)\big)}{\log(1+\epsilon)}$.
\end{proposition}
\begin{proof}
Let $e_0=(u_0,v_0)$ and observe that at the end of the visit of $e_0$ and $v_0$
\begin{align*}
{\tt last}(v_0)+d_T(v_0,t)	& =d_{G-e_0}(s,v_0)+d_T(v_0,t) \\
				& \leq d_{G-e_0}(s,t)+d_T(t,v_0)+d_T(v_0,t)\\
				& \leq d_{G-e_0}(s,t)+2d_T(s,t) = d_{G-e_0}(s,t)+2d_G(s,t).
\end{align*}
Since $d_{G-e_0}(s,t) \in S$, $\sqrt{1+\epsilon}\cdot d_{G-e_0}(s,t)<{\tt last}(v_0)+d_T(v_0,t)$,  and therefore
\begin{equation}\label{eq:one}
d_{G-e_0}(s,t) < \frac{2}{\sqrt{1+\epsilon}-1} d_G(s,t).
\end{equation}

Next, observe that the value ${\tt last}(t)$ at the beginning of the visit of edge $e_i$, with $1\leq i \le k$, is equal to $d_{G-e_{i-1}}(s,t)$. Since $d_{G-e_i}(s,t) \in S$, we have that
\begin{equation}\label{eq:two}
\sqrt{1+\epsilon}\cdot d_{G-e_i}(s,t)<d_{G-e_{i-1}}(s,t)\,\,\,\,\,\,\,\, \text{for every $i=1,\ldots,k$}.
\end{equation}
Thus, if, for any $i>0$, we combine inequality (\ref{eq:one}) and all the inequalities (\ref{eq:two}) with $j\leq i$, we obtain $(1+\epsilon)^{i/2}d_{G-e_i}(s,t)<2/(\sqrt{1+\epsilon}-1)d_G(s,t)$, i.e.,
\[
d_{G-e_i}(s,t) < \frac{2}{(\sqrt{1+\epsilon}-1)(1+\epsilon)^{i/2}}d_G(s,t).
\]

Moreover, using $d_{G}(s,t)\leq d_{G-e_k}(s,t)$ in $d_{G-e_k}(s,t)<2/\big((\sqrt{1+\epsilon}-1)(1+\epsilon)^{k/2}\big)d_G(s,t)$ we obtain $(1+\epsilon)^{k/2}<2/(\sqrt{1+\epsilon}-1)$, i.e.,
\[
k<2\cdot \frac{\log\big(2/(\sqrt{1+\epsilon}-1)\big)}{\log(1+\epsilon)}.
\]
The claim follows.
\end{proof}

Observe that $\log \big(2/(\sqrt{1+\epsilon}-1)\big)=O(\log(1/\epsilon))$, and that $\log(1+\epsilon)=\Theta(\epsilon)$. Therefore, using Proposition \ref{prop:rank_of_a_node} and the fact that $|S'|=n-1$, we obtain

\begin{corollary}\label{cor:size_of_type_1_paths}
$|S\cup S'|=O\left(n\cdot\frac{1}{\epsilon} \log \frac{1}{\epsilon}\right)$.
\end{corollary}

\begin{lemma}
Algorithm \ref{alg:preprocessing} can be implemented to run in $O(mn+n^2\log n)$ time and \linebreak $O\left(m+n\cdot\frac{1}{\epsilon} \log \frac{1}{\epsilon}\right)$ space.
\end{lemma}
\begin{proof}
First we prove the time bound. Clearly, the inizialization of variables takes $O(n)$ time. Let $e=(u,v)$ be an edge that is visited by the algorithm. The algorithm computes an SPT of $G-e$ rooted at $s$ in $O(m+n\log n)$ time. Let $t \neq v$ be the vertex that is going to be visited by the algorithm and let $t'$ be the parent of $t$ in $T$. Observe that
\begin{align}\label{eq:efficient_algorithm}
\min_{z \in A(t,e)}&\big\{{\tt last}(z) +d_T(z,t)\big\} = \min\left\{{\tt last}(t),\min_{z \in A(t',e)}\big\{{\tt last}(z)+d_T(z,t)\big\}\right\}\notag{}\\
		& = \min\left\{{\tt last}(t), \min_{z \in A(t',e)}\big\{{\tt last}(z)+d_T(z,t')\big\}+w(t',t)\right\}\\
		& = \min\big\{{\tt last}(t),{\tt dist}(t',e)+w(t',t)\big\},\notag{}
\end{align}
Therefore, each value ${\tt dist}(t,e)$ can be computed in constant time rather than in $O(n)$ time. Hence, the overall running time is $O(mn+n^2\log n)$.

Concerning the space complexity, observe that, from Equation (\ref{eq:efficient_algorithm}), the algorithm does not need to store all the values ${\tt dist}(t,e)$ but, for each $t$, it is enough to remember the last computed value ${\tt dist}(t,e)$. This can be clearly done with an array of $n$ elements. Next, observe that, during the visit of $e$, the algorithm only needs the one-to-all distances in $G-e$. This implies that there is no need to keep all the $n-1$ SPT's of $G-e$, for every $e \in E(T)$, at the same time and therefore, all these SPT's can share the same $O(n)$ space. Finally, $|S\cup S'|=O\left(n\cdot\frac{1}{\epsilon} \log \frac{1}{\epsilon}\right)$ by Corollary \ref{cor:size_of_type_1_paths}. The claim follows.
\end{proof}

\subsection{The data structure.}

We now describe how the values in $S$ and $S'$ can be organized in a data structure of size $O\left(n\cdot \frac{1}{\epsilon} \log \frac{1}{\epsilon}\right)$ so that our distance oracle can compute a $(\sqrt{1+\epsilon})$-approximation of ${\tt dist}(t,e)$ in $O\left(\log n\cdot \frac{1}{\epsilon} \log \frac{1}{\epsilon}\right)$ time.

Remind that we say that $e'\prec e''$ if the preprocessing algorithm has visited $e'$ before visiting $e''$, and that we also use $e'\preceq e''$ to denote that either $e'\prec e''$ or $e'=e''$. Let $k=\left\lfloor 2\cdot \frac{\log\big(2/(\sqrt{1+\epsilon}-1)\big)}{\log(1+\epsilon)}\right\rfloor$ and let $a_i=\frac{2}{(\sqrt{1+\epsilon}-1)(1+\epsilon)^{i/2}}$. Finally, for every $i=0,1,\ldots,k$, let
\[
S_i=\Big\{d_{G-e'}(s,z) \in S \mid a_{i+1}\cdot d_G(s,z)\leq d_{G-e'}(s,z)< a_i \cdot d_G(s,z)\Big\}.
\]
By Proposition \ref{prop:rank_of_a_node}, we have that $\{S_i\mid i=0,1,\ldots,k\}$ is a partition of $S$. 

We maintain a set of $k+1$ trees ${\cal T}_0,{\cal T}_1,\ldots,{\cal T}_{k}$, one for each $S_i$.
Each tree ${\cal T}_i$ is a copy of $T$, where each vertex $z$, such that $d_{G-e'}(s,z)\in S_i$, has a label $\ell_i(z)=e'$. Every other vertex $z \in V(G)\setminus S_i$ has a label $\ell_i(z)=\infty$ such that $e'\prec \infty$, for every edge $e' \in E(T)$.


In the following, we denote the value of ${\tt last}(z)$ at the end of the visit of edge $e'$ by ${\tt last}(z,e')$. First of all, we prove the following proposition.

\begin{figure}[t]
	\centering
	\includegraphics[scale=1]{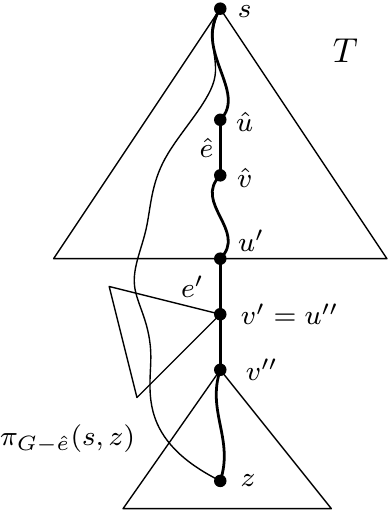}
	\caption{Representation of the proof of Proposition~\ref{prop:labels}. The shortest path between $s$ and $t$ in $T$ is shown. Notice that the path $\pi_{G-\hat{e}}(s,z)$ is edge disjoint from the path $\pi_T(\hat{v},z)$.}
	\label{fig:labels}
\end{figure}

\begin{proposition}\label{prop:labels}
If $e'\preceq e''$, then ${\tt last}(z,e'')\leq {\tt last}(z,e')$.
\end{proposition}
\begin{proof}
Let $e'=(u',v')$ and $e''=(u'',v'')$. Notice that the claim can be proved by showing that it holds under the assumption that $v'=u''$. Furthermore, we can also assume that ${\tt last}(z,e')\neq \infty$ as well as ${\tt last}(z,e'')\neq {\tt last}(z,e')$, otherwise the claim would be trivially true. This last assumption together with $v'=u''$ imply that ${\tt last}(z,e'')=d_{G-e''}(s,z)$. Let ${\tt last}(z,e')=d_{G-\hat e}(s,z)$, for some $\hat e\preceq e'$, with $\hat e=(\hat u,\hat v)$. Clearly, $d_{G-\hat e}(s,z) \in S\cup S'$. Therefore, by Lemma \ref{lm:disjoint_paths}, $\pi_{G-\hat e}(s,z)$ and $\pi_T(\hat v,z)$ are edge disjoint (see Figure~\ref{fig:labels}). Since $e''$ is an edge of $\pi_T(\hat v,z)$, $\pi_{G-\hat e}(s,z)$ is also a path from $s$ to $t$ in $G-e''$ and therefore
${\tt last}(z,e'')=d_{G-e''}(s,z)\leq d_{G-\hat e}(s,z)={\tt last}(z,e')$.
\end{proof}

Let $e=(u,v) \in E(T)$ and let $t$ be a vertex of $T_v$. Using Proposition \ref{prop:labels}, we have that either ${\tt dist}(t,e)={\tt last}(v,e)+d_T(v,t)=d_{G-e}(s,v)+d_T(v,t)$, or
\begin{flalign*}
{\tt dist}(t,e) & = \min\big\{{\tt last}(z,e)+d_T(z,t)\mid z \in A(t,e)\setminus\{v\}\big\}\\
				& = \min\big\{{\tt last}(z,e)+d_T(z,t)\mid z \in A(t,e)\big\}\\
				& = \min_{i=0,1,\ldots,k}\Big\{\min\big\{{\tt last}(z,\ell_i(z))+d_T(z,t) \mid z \in A(t,e)\wedge\ell_i(z)\preceq e\big\}\Big\}\\
				& = \min_{i=0,1,\ldots,k}\Big\{\delta_i:=\min \big\{d_{G-e'}(s,z)+d_T(z,t) \\
				&& \mathllap{ \mid z \in A(t,e) \wedge d_{G-e'}(s,z) \in S_i \wedge  e'\preceq e\big\}\Big\}. \;\; }
\end{flalign*}

In the former case, ${\tt dist}(t,e)$ is available in $O(1)$ time, since $d_{G-e}(s,v)$ is stored in $S'$. In the latter case, we now show how to compute, for any fixed $i=0,1,\ldots,k$, a $(\sqrt{1+\epsilon})$-approximate upper bound to $\delta_i$
in $O(\log n)$ time. Using Proposition \ref{prop:rank_of_a_node}, this will imply that our oracle is able to answer a query in $O\left(\log n\cdot \frac{1}{\epsilon} \log \frac{1}{\epsilon}\right)$ time.

First of all, we prove that the labels of each ${\cal T}_i$ satisfy a nice property.

\begin{lemma}\label{lm:label_monotonicity}
Let $z'$ and $z''$ be two distinct vertices of $A(t,e)$ such that $z''$ is a proper ancestor of $z'$ and $\ell_i(z')=e'$ and $\ell_i(z'')=e''$, for some edges $e',e'' \in E(T)$, with $e',e''\preceq e$ (possibly, $e'=e''$).
We have that $d_{G-e''}(s,z'')+d_T(z'',t)\leq \sqrt{1+\epsilon}\cdot\big( d_{G-e'}(s,z')+d_T(z',t)\big)$.
\end{lemma}
\begin{proof}

Since $d_{G-e''}(s,z'') \in S_i$, we have that $d_{G-e''}(s,z'')< a_i \cdot d_G(s,z'')$. Furthermore, $d_{G-e'}(s,z') \in S_i$ implies that $d_{G-e'}(s,z') \geq a_{i+1}\cdot d_G(s,z')=a_i/ \sqrt{1+\epsilon}\cdot d_G(s,z')$. As a consequence, $d_{G-e''}(s,z'')+d_T(z'' ,t)<a_i \cdot d_G(s,z'')+d_T(z'',z')+d_T(z',t)\leq a_i \cdot d_G(s,z')+d_T(z',t)
 \leq\sqrt{1+\epsilon}\cdot d_{G-e'}(s,z')+d_T(z',t)\leq \sqrt{1+\epsilon}\cdot\big(d_{G-e'}(s,z')+d_T(z',t)\big)$.
\end{proof}

Let $z'' \in A(t,e)$ be the vertex closest to $v$ w.r.t. $T$ such that $\ell_i(z)=e''\preceq e$, if such a vertex exist. Let $\delta_i=d_{G-e'}(s,z')+d_T(z',t)$, for some $e'$ and $z'$ such that $z' \in A(t,e)$, $d_{G-e'}(s,z') \in S_i$, and $e' \preceq e$. Observe that $d_{G-e''}(s,z'')+d_T(z'',t)\geq \delta_i$. Moreover, since $z'$ and $z''$ satisfy all the hyphotesis of Lemma \ref{lm:label_monotonicity}, we have that

\[\delta_i\leq d_{G-e''}(s,z'')+d_T(z'',t)\leq \sqrt{1+\epsilon}\cdot \delta_i.\]

Therefore, the value $d_{G-e''}(s,z'')+d_T(z'',t)$ is a $(\sqrt{1+\epsilon})$-approximate upper bound to the value $\delta_i$. Now we show how the vertex $z''$ can be computed in $O(\log n)$ time.

To this aim, we preprocess each tree $\mathcal{T}_i$ in order to build a linear-size data structure that answers \texttt{BVQ} queries in constant time. This can be done in $O(n \log n)$ time per tree. We also preprocess $T$ so we are able to perform \emph{level-ancestor} queries in constant time. The size needed by this latter data structure is $O(n)$ and it can be built in linear-time \cite{BerkmanV94,BenderF04}.
In a level ancestor query, we are given a vertex $x \in V(T)$ and a positive integer $h$, and we ask for the ancestor $y$ of $x$ such that $\pi_T(x,y)$ contains exactly $h$ edges. We can then find $z''$ by performing a binary search over the vertices of $A(t,e)$, as follows.

Let $e=(u,v)$, we perform a level ancestor query on $T$ to find the vertex $x$ of $\pi_T(v,t)$ that divides the path into roughly two halves. Let $x'$ be the parent of $x$, and let $y$ and $y'$ be the vertices of $\pi_T(x,t)$ and $\pi_T(v,x')$ of minimum labels, respectively. Notice that $y$ and $y'$ can be found in constant time by performing two \texttt{BVQ} queries on $\mathcal{T}_i$. If $\ell_i(y')\preceq e$, then we remember $y'$ as the best vertex found so far and we iterate the binary search in $\pi_T(v,x')$. Otherwise, if $e\prec\ell_i(y')$, then we compare $\ell_i(y)$ and $e$. If $\ell_i(y) \preceq e$, then we remember $y$ as the best vertex found so far and we iterate the binary search in $\pi_T(x,t)$. If $e \prec \ell_i(y)$, then we can complete our binary search and return the best vertex found, if any.

We have then proven the following:

\begin{theorem}
Let $G$ be a non-negatively real weighted and undirected $n$-vertex and $m$-edge graph, and let $s$ be a source node.
For any arbitrarily small $0 < \epsilon <1$, there exists an EFT $(1+\epsilon)$-SSDO that has size $O\left(n\cdot \frac{1}{\epsilon} \log \frac{1}{\epsilon}\right)$ and $O\left(\log n\cdot \frac{1}{\epsilon} \log \frac{1}{\epsilon}\right)$ query time, and that can be constructed using $O(mn+n^2 \log n)$ time and $O\left(m+n\cdot \frac{1}{\epsilon} \log \frac{1}{\epsilon}\right)$ space.
\end{theorem}

\section{Lower bounds on the size of additive EFT ASPT and SSDO}
In this section, we give a lower bound on the bit size of an EFT $\beta(d)$-additive SSDO. Recall that after the failure of any edge, such an oracle must return an estimation $d'$ of the actual distance $d$ between $s$ and any node such that $d \leq d' \leq d+\beta(d)$, where $\beta$ is any positive real function. We are able to prove the following result:

\begin{theorem}\label{thm:lower_bound}
		Let $\beta(d)=k d^{1-\delta}$, for arbitrary $k \ge 1$ and $0 < \delta \le 1$. Then, there exist classes of polynomially weighted graphs with $n$ nodes such that:
	\begin{enumerate}
\item any \eft $\beta(d)$-additive ASPT has $\Omega(n^2)$ edges;
\item any \eft $\beta(d)$-additive SSDO has $\Omega(n^2)$ bit size for at least an input graph, regardless of its query time.
\end{enumerate}	
\end{theorem}
\begin{proof}
		We first discuss the lower bound on the size of any single-source \eft $\beta(d)$-additive spanner.
		Our construction is inspired by the one given in \cite{PP13} for EFT ASPTs on unweighted graphs.
		Consider a graph $G$ similar to the one shown in Figure~\ref{fig:n2_lower_bound}, consisting of:

\begin{figure}[ht]
	\centering
	\includegraphics[scale=1]{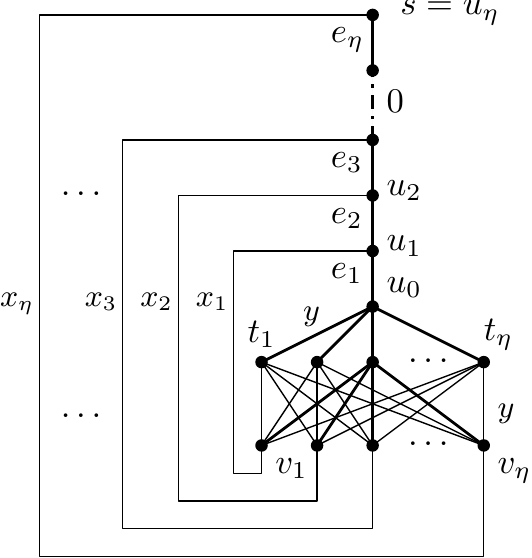}
	\caption{The graph $G$ used to show that a single-source \eft $\beta(d)$-additive spanner requires $\Omega(n^2)$ edges.}
	\label{fig:n2_lower_bound}
\end{figure}

		\begin{itemize}
			\item a path of $\eta+1$ vertices $\langle s = u_\eta, u_{\eta-1}, \dots, u_0 \rangle$, with $\eta \ge k+1$. We let $e_i = (u_i, u_{i-1})$ and we set $w(e_i) = 0$;
			
			\item a complete bipartite graph of $2\eta$ vertices, namely $t_1, \dots, t_\eta$ and $v_1, \dots, v_\eta$. Each edge $(t_i, v_i)$  has weight $y \ge 1$;
			
			\item a star connecting $u_0$ to the vertices in $\{ t_i : 1 \le i \le \eta \}$ whose edges have weight $y \ge 1$;
			
			\item and an edge $(u_i, v_i)$ of weight $x_i$, for every  $1 \le i \le \eta$.
		\end{itemize}

We will show how to set the weights $x_i$s and $y$ so that the only single-source \eft $\beta(d)$-additive spanner of $G$ is $G$ itself. First of all, we will choose $2y < x_1 < x_2 < \dots < x_\eta$. As a consequence, we have that $d_G(s, t_i ) = y$, and $d_G(s, v_i ) = 2y$, for every $i$. Therefore, the SPT of $G$ is similar to that shown in bold in Figure~\ref{fig:n2_lower_bound}.

The idea of the construction is that of suitably balancing the values $x_i$s and $y$ so that each edge of the bipartite graph will belong to every spanner of $G$. In particular, we would guarantee that when $e_i$ fails, the shortest path from $s$ to any vertex $t_h$ will be the path $\pi$ of weight $x_i + y$ passing through the edges $(u_i, v_i)$ and $(v_i, t_h)$. Moreover, any other path $\pi'$ towards $t_h$ will have a weight larger than $x_i + y + \beta(x_i + y)$. If $\pi' \neq \pi$, either $\pi'$ traverses an edge $(u_j, v_j)$ with $j>i$ or it traverses $(u_i, v_i)$ but not the edge $(v_i, t_h)$, hence it must contain at least three edges of weight $y$. We want $\pi'$ to be larger than $\pi$ by an additive term of $\beta(x_i+y)$. This is challenging due to the fact that if $y$ is too large then the path of length $x_j+y$ (with $j> i$) is comparable with the distance $x_i+y$, while if $y$ is too small then the path of length $x_i+3y$ could be good enough. Therefore, for each $i$, we require the following conditions to hold:

	\begin{gather*}
			\begin{cases}
				x_j + y > x_i + y + \beta(x_i + y) & \forall i \, \forall j > i\\
				x_i + 3y > x_i + y + \beta(x_i + y) & \forall i
			\end{cases} \\
			\Updownarrow\\
			\begin{cases}
				x_{i+1} + y > x_i + y + \beta(x_i + y) & \forall i < \eta \\
				2y > \beta(x_i + y) & \forall i
			\end{cases}
	\end{gather*}
	
	\noindent where we used the fact that the values $x_i$s are strictly increasing.
	
	Let $0 < \gamma \le 1$, and let
 $x_1 = 2y + \gamma$, and $x_{i+1} = x_i + \beta(x_i + y) + \gamma$. Notice that the first set of equations is now trivially satisfied. Moreover, in the second set of equations, since $\beta(\cdot)$ is an increasing function, the latter equation (i.e., the equation for $i=\eta$) implies all the others. Hence, the whole system reduces to:	
	\begin{equation}
		\label{eq:condition_x}
		x_\eta + 3y > x_\eta + y + \beta(x_\eta + y) \iff 2y > \beta(x_\eta + y)	
	\end{equation}

By defining $z_i = x_i + y$, we have $z_1 = 3y + \gamma$, $z_{i+1} = z_i + \beta(z_i) + \gamma$ and \eqref{eq:condition_x} becomes:
	\begin{equation}	
		\label{eq:condition_z}
		2y > \beta(z_\eta)
	\end{equation}
	
We now prove an upper bound to the value of $z_\eta$s:
	\begin{lemma}
		$z_\eta \le 4y (2n)^{k + 2 + 2/\delta}$ for any choice of $0<\gamma \le 1$.
	\end{lemma}
	\begin{proof}
		Consider the sequence $\langle z_1, z_2, \dots, z_i, \dots, z_\eta \rangle$ and let $\tau \in \{1, \dots, \eta \}$ be the largest index $i$ such that $k z_i^{1-\delta} + \gamma > \frac{k}{i} z_i$, if it exists.
		
		If such an index $\tau$ exists then, for $i > \tau$, we have $k z_i^{1-\delta} + \gamma \le \frac{k}{i} z_i$ and hence:
		\[
			z_{i+1} = z_i + \beta(z_i) + \gamma = z_i + k z_i^{1-\delta} + \gamma  \le z_i + \frac{k}{i} z_i = \left(1 + \frac{k}{i} \right) z_i
		\]
		from which we can easily get:
		\begin{multline*}		
			z_\eta \le  z_{\tau+1} \!\! \prod_{i=\tau+1}^{\eta-1} \!\!\! \left( \! 1 + \frac{k}{i} \right)
			\! \le z_{\tau+1} \prod_{i=1}^{\eta-1} \frac{\lceil k \rceil + i}{i}
			= \frac{z_{\tau+1}}{\lceil k \rceil !} \!\! \prod_{i=0}^{\lceil k \rceil - 1} (\eta + i) \\
			\le z_{\tau+1} (\eta + k)^{\lceil k \rceil} \!\le z_{\tau+1} (2\eta)^{k+1}
		\end{multline*}

		To bound $z_{\tau+1}$ we use the fact that $k z_\tau^{1-\delta} + \gamma > \frac{k}{\tau} z_\tau \implies
		(k+1) z_\tau^{1-\delta} > \frac{k}{\tau} z_\tau \implies \left( \frac{k+1}{k} \tau \right)^{1/\delta} > z_\tau$ to write:
		\begin{multline*}
			z_{\tau+1} = z_\tau + \beta(z_\tau) + \gamma = z_\tau + k z_\tau^{1-\delta} + \gamma
			\le (k+1) z_\tau + \gamma
			\le (k+2) z_\tau \\
			< (k+2) \left( \frac{k+1}{k} \tau \right)^{1/\delta} 
			\le (k+2)^{1+\frac{1}{\delta}} \eta^{1/\delta}.
		\end{multline*}
		Combining the previous two inequalities, we obtain the claim:
		\[
			z_\eta  \le z_{\tau+1 }(2 \eta )^{k+1} < (k+2)^{1+1/\delta} \eta^{1/\delta} (2\eta)^{k+1}
			< (2\eta)^{k+2+2/\delta}
			< 4y (2n)^{k + 2 + 2/\delta}.
		\]
				
		\noindent If no such index $\tau$ exists, then $k z_i^{1-\delta} + \gamma \le \frac{k}{i} z_i$ for every $i=1,\dots,\eta$ and, with a similar argument, we have:
		\begin{align*}		
			z_\eta & \le z_1 \prod_{i=1}^{\eta-1} \left(1 + \frac{k}{i} \right) \le z_1 (2\eta)^{k+1}  = (3y+\gamma) (2\eta)^{k+1} \\
			& < 4y (2n)^{k+1}  < 4y (2\eta)^{k + 2 + 2/\delta}.
		\end{align*}
	\end{proof}		
\noindent Therefore, \eqref{eq:condition_z} is satisfied for a large enough value of $y$, indeed:
\begin{multline*}
	2y > \beta( 4y (2n)^{k + 2 + 2/\delta} ) = y^{1-\delta} k \left( 4 (2n)^{k + 2 + 2/\delta} \right)^{1-\delta}
	\iff \\
	y^\delta > \frac{k}{2} \left( 4 (2n)^{k + 2 + 2/\delta} \right)^{1-\delta}
	\iff
	y > \left( \frac{k}{2} \right)^\frac{1}{\delta} \left(  4 (2n)^{k + 2 + 2/\delta} \right)^{\frac{1}{\delta}-1}.
\end{multline*}
\noindent Notice that both $y$ and each $x_i$ are at most $n^c$ for a suitable constant $c$ depending on $k$ and $\delta$. This means that $O(\log n)$ bits suffice to encode the weight of any edge.

To conclude the proof, we now discuss the lower bound of $\Omega(n^2)$ bits on the size of any \eft $\beta(d)$-additive SSDO. We will use an argument similar to the one shown in \cite{TZ05}.
To this respect, consider the set $\mathcal{G}$ of all the subgraphs obtained from $G$ by removing any subset of edges belonging to the bipartite graph. Notice that $|\mathcal{G}| = 2^{\Theta(n^2)}$. For each graph $G_j \in \mathcal{G}$, let $\mathcal{O}_j$ be the corresponding \eft $\beta(d)$-additive SSDO. Given two different graphs $G_j$ and $G_h$, we now prove that $\mathcal{O}_j$ must differ from $\mathcal{O}_h$.
Indeed, w.l.o.g., let $(v_i, t)$ be an edge in $E(G_j) \setminus E(G_h)$.
We query the distance from $s$ to $t$ when edge $e_i$ is failing. Oracle $\mathcal{O}_j$ must return a distance between $x_i + y$ and $x_i + y + \beta(x_i + y)$, while oracle $\mathcal{O}_h$ must return a distance $d'$ that is at least the minimum between $x_j + y$, for some $j > i$, and $x_i + 3y$. As we have proved that $d' > x_i + y + \beta(x_i + y)$, we have that the two oracles must return different answers, and hence they differ.
Since there are $2^{\Theta(n^2)}$ distinct oracles, at least one of them must have a size of at least $\log 2^{\Theta(n^2)} = \Theta(n^2)$.
\end{proof}

\bibliographystyle{plainurl}
\bibliography{biblio}

\end{document}